\documentclass[runningheads]{llncs}
\usepackage[T1]{fontenc}   
\usepackage{amsmath}
\usepackage{amssymb}
\usepackage{graphicx}
\newtheorem{fact}{Fact}
\newcommand{\junk}[1]{}

\newlength{\boxwidth}
\begin{document}





\title{Multiplication of 0-1 matrices via clustering\thanks{A
preliminary version of this article appeared in
\emph{Proceedings of Frontiers of Algorithmics - the 19th International Joint
Conference} (IJTCS-FAW~2025),
Lecture Notes in Computer Science, Vol.~15828, pp.~92--102,
Springer~Nature, 2025.}}
\author{
  Jesper Jansson
  \inst{1}
\and
Miros{\l}aw Kowaluk
\inst{2}
\and
Andrzej Lingas
\inst{3}
\and
  Mia Persson
  \inst{4}}
\institute{
Graduate School of Informatics, Kyoto University, Kyoto, Japan.
\email{jj@i.kyoto-u.ac.jp}
\and  
  Institute of Informatics, University of Warsaw, Warsaw, Poland.
  \texttt{kowaluk@mimuw.edu.pl}
\and
  Department of Computer Science, Lund University,
Lund, Sweden.
\texttt{Andrzej.Lingas@cs.lth.se}
\and
Department of Computer Science and Media Technology, Malm\"o University, Malm\"o, Sweden.
\texttt{Mia.Persson@mau.se}
}
\pagestyle{plain}

\maketitle


\begin{abstract}
We study applications of clustering (in particular, the
Hamming $k$-center
clustering problem) in the design of efficient and practical
algorithms for computing an approximate and the exact arithmetic
matrix product of two 0-1 rectangular matrices with clustered rows or
columns, respectively. Our results in part can be regarded as an extension of the clustering-based approach to Boolean square matrix multiplication due to Arslan and Chidri (CSC 2011).
We provide a simple and efficient deterministic algorithm for
approximate matrix product of 0-1 matrices, where the additive error
is proportional to the minimum maximum radius in an $\ell$-center
clustering of the rows of the first matrix or an $k$-center clustering
of the columns of the second matrix. We use the approximation
algorithm as a preprocessing after which a query asking for the exact
value of an arbitrary  entry in the product matrix can be answered in time
proportional to the additive error.  As a consequence, we obtain a
simple deterministic algorithm for the exact matrix product of 0-1
matrices.
We also present
  an alternative simple deterministic algorithm for the
exact product and in addition, faster analogous randomized algorithms
for an approximate and the exact matrix products of 0-1 matrices based
on randomized $\ell$- and $k$-center clustering.
\junk{  \footnote{A preliminary version of this article appeared
    in the proceedings of
    Frontiers in Algorithmics - 19th International
    Joint Conference, IJTCS-FAW 2025,
    Lecture Notes in Computer Science 15828, pp. 92-102.}}
\end{abstract}




\section{Introduction}
The arithmetic matrix product of two 0-1 matrices is closely related
to the Boolean one of the corresponding Boolean matrices.  Both are
basic tools in science and engineering.  For square $n\times n$
matrices, both can be computed in $O(n^{2.372})$ time by using 
fast matrix multiplication algorithms based on algebra
\cite{ADV25,VXZ24}.  Unfortunately, the latter algorithms suffer from
huge overheads and no truly subcubic practical algorithms for any of
these two problems is known.  Therefore, many researchers studied the
complexity of these products for special input matrices, e.g., sparse
or structured matrices \cite{AC10,Emi24,BL,FJLL18,GL03,P13,YZ}, providing
faster and often more practical algorithms.

The method of multiplying matrices
with clustered rows or columns, proposed for Boolean matrix
product in \cite{BL} and subsequently
generalized in \cite{FJLL18,GL03} and used in
\cite{AMBFGR_25}, relies on the construction of an
approximate minimum spanning tree of the rows of the first input matrix or the
columns of the second input matrix in a Hamming space.
Then, each column or each row of the product matrix
is computed with the help of a traversal of the tree in time
proportional to the total Hamming cost of the tree up to a logarithmic
factor. Simply, the next entry in a column or a row in the product
matrix can be obtained from the previous one in time roughly
proportional to the Hamming distance between the consecutive
(in the tree traversal) corresponding rows or
columns of the first or the second input matrix, respectively.  Thus, in
case the entire tree cost is substantially subquadratic in $n,$ the
total running time of this method becomes substantially subcubic
provided that a good approximation of a minimum spanning tree of the
rows of the first input matrix or the columns of the second one can be
constructed in substantially subcubic time.  As for simplicity and
practicality, a weak point of this method is that in order to
construct such an approximation relatively quickly, it employs a
rather involved randomized algorithm.

Arslan and Chidri presented a simple algorithm for the Boolean
matrix product of two input Boolean matrices with rows (in the first matrix)
and columns (in the other matrix) within a small diameter $d$ (in
terms of the Hamming distance) in \cite{AC11}. The basic idea is
to compute the inner Boolean product of a representative of the rows
of the first matrix and a representative of columns of the second matrix
and then relatively cheaply recover the entries of the Boolean product
matrix from the inner product. Their algorithm runs in $O(dn^2)$
time for $n\times n$ matrices.
They also presented a generalization of their algorithm to include
grouping the rows of the first matrix as well as the columns of the second
matrix into $k$ clusters of maximum radius within $2$ of the minimum  $s.$
They incorporate Gonzalez's $2$-approximation algorithm for
the $k$-center problem \cite{Gon85} to achieve such a $k$-clustering.
Their general
algorithm runs in $O((k+s)n^2)$ time\footnote{They also claim
  that one can speed up their general algorithm by
  using the fast $2$-approximation algorithm for the $k$-center
  problem due to Feder and Greene \cite{FG88} instead of Gonzalez's one.
  Unfortunately, the former algorithm seems to have rather
  an exponential hidden dependence on the dimension (in this case $n$)
  in view of Lemmata 4.2 and 4.3 in \cite{FG88}
  than a linear one as it is assumed in \cite{AC11}.}.
  
In case of the arithmetic matrix product of 0-1 matrices, in some
cases, a faster approximate arithmetic matrix multiplication can be
more useful \cite{CL99,P13}.  Among other things, it can enable to
identify largest entries in the product matrix and it can
also be used to provide a fast estimation of the number of: the so called
witnesses for the Boolean product of two Boolean matrices
\cite{GKL09}, triangles in a graph, or more generally, subgraphs
isomorphic to a small pattern graph \cite{FKL15} etc.  There is a
number of results on approximate arithmetic matrix multiplication,
where the quality of approximation is expressed in terms of the
Frobenius matrix norm $||\ ||_F$ (i.e., the square root of the sum of
the squares of the entries of the matrix) \cite{CL99,P13}.

Cohen and Lewis \cite{CL99} and Drineas {\em et al.} \cite{DKM06}
used random sampling
  to approximate arithmetic matrix product. Their
  articles provide an approximation $D$
  of the matrix product $AB$ of two $n\times n$ matrices $A$ and $B$ such
  that $||AB-D||_F=O(||AB||_F/\sqrt c)$, for a parameter $c>1$ (see
  also \cite{P13}).
  The approximation algorithm in \cite{DKM06} runs in  $O(n^2c)$ time.
  Drineas et al. \cite{DKM06} also derived bounds on
  the entrywise differences
  between the exact matrix product and its approximation.
  Unfortunately, the best of these
  bounds is $\Omega(M^2n/\sqrt c)$, where $M$ is the maximum value
  of an entry in $A$ and $B.$
  By using a sketch
  technique, Sarl\'os \cite{S06}
  obtained
  the same Frobenius norm guarantees,
also in $O(n^2c)$ time. However, he derived stronger individual
upper bounds on the additive error of each entry $D_{ij}$ of the
approximation matrix $D$. They are of
the form $O(||A_{i*}||_2||B_{*j}||_2/\sqrt c)$,
where $A_{i*}$ and $B_{*j}$ stands for the $i$-th row of $A$
and the $j$-th column of $B,$ respectively, that hold with
high probability.
More recently, Pagh \cite{P13} presented a
randomized approximation $\tilde{O}(n(n+c))$-time
algorithm for the arithmetic product
of $n\times n$ matrices $A$ and $B$ 
such that each entry of the approximate
matrix product differs at most by
$||AB||_F/\sqrt c$ from the correct one. His algorithm first compresses
the matrix product to a product of two polynomials and then uses the
fast Fourier transform  to multiply the polynomials.
Subsequently, Kutzkov \cite{K13} developed analogous deterministic algorithms
employing different techniques. For approximation results related to sparse
arithmetic matrix products, see \cite{IS09,P13}.

\subsection{Our contributions}

In this article, we exploit the possibility of applying the classic simple
$2$-approxi\-mation algorithm for the $k$ center clustering
problem \cite{Gon85} in order to derive efficient and practical
deterministic algorithms for computing an approximate and the exact
arithmetic matrix product of two 0-1 rectangular matrices $A$ and $B$
with clustered rows or columns, respectively. In addition, we consider
the possibility of applying
the recent faster randomized algorithm for $k$ center clustering from
\cite{KLP26} instead of that from \cite{Gon85}.

Most of our results are  based on the idea of using cluster
representatives / centers to compute faster a smaller matrix product and
then relatively cheaply updating the small product to the target one.
This natural idea has been used already in \cite{AC11}\footnote{We
  came across this idea independently without knowing \cite{AC11} before
  writing the preliminary conference version of this
  article.}. Our approach conceptually extends
that in \cite{AC11} in several ways.  We consider the more general
problem of computing the arithmetic matrix product of 0-1 rectangular
matrices while \cite{AC11} is concerned with the Boolean matrix
product of square Boolean matrices. We also provide clustering-based
approximation and preprocessing for a query
on the value of single entry of the aforementioned
arithmetic product; no prior clustering-based
explicit results of this kind seem to be
known in the literature.
In \cite{AC11}, the rows of the first matrix are
grouped in the same number of clusters as the columns of the second matrix,
while in our case the numbers are in general different. Also, in
\cite{AC11} representatives/centers of both row clusters and column
clusters are used to compute the smaller product while we use
in the basic deterministic version
either the row ones or the column ones.
We also provide a faster randomized cluster-based method
  of multiplying 0-1 matrices relying on a fast randomized
  clustering.

The $k$-center clustering problem in a Hamming space $\{0,1\}^d$ is
for a set $P$ of $n$ points in $\{0,1\}^d$ to find a set $T$ of $k$
points in $\{0,1\}^d$ that minimize
$\max_{v\in P} \min_{u\in T} \mathrm{ham}(v,u),$ where
$ \mathrm{ham}(v,u)$ stands for the Hamming distance between $v$ and
$u$ (i.e., the number of coordinates where $v$ and $u$ differ).
Each center in $T$ induces a cluster consisting of all points in $P$ for which
it is the nearest center.

Let $\lambda(A,\ell,row)$ and $\lambda(B,k,col)$ denote the minimum
maximum Hamming distance between a row of $A$ or a column of $B$ and
the closest center in an $\ell$-center clustering of the rows of $A$
or in a $k$-center clustering of the columns of $B,$ respectively.
(Note that in extreme cases, $\lambda(A,\ell,row)$
or $\lambda(B,k,col)$
can even be equal to $0.$) In
particular, when $A$ and $B$ are square matrices of size $n\times n$,
we obtain the following results.
 \begin{enumerate}
 \item
   A simple deterministic algorithm that approximates each entry of
   the arithmetic matrix product of $A$ and $B$ within an additive
   error of at most $2\lambda(A,\ell,row)$ in $O(n^2\ell)$ time or at
   most $2\lambda(B,k,col)$ in $O(n^2k)$ time.
   Similarly, a simple randomized algorithm that for any
   $\epsilon \in (0,\frac 12)$ approximates each entry of the product
   within an additive error of with high probability (w.h.p.) at most
   $(2+\epsilon)(\lambda(A,\ell,row)+\lambda(B,k,col))$ in time
   $O(n^2\log n /\epsilon^2+\ell n k)$.
 \item
   A simple deterministic preprocessing of the matrices $A$ and $B$ in
   $O(n^2\ell)$ time or $O(n^2k)$ time after which every query asking
   for the exact value of an arbitrary entry of the arithmetic matrix
   product of $A$ and $B$ can be answered in
   $O(\lambda(A,\ell,row)+1)$
   time or
   $O(\lambda(B,k,col)+1)$ time, respectively.  Similarly, a
   simple randomized preprocessing of the matrices in time
   $O(n^2\log n + \ell n k)$ after which every query asking
   for the exact value of an arbitrary entry of the product of the
   matrices can be answered in
   $O(\lambda(A,\ell,row)+\lambda(B,k,col)+1)$ time w.h.p.
        \item
       Simple deterministic algorithms for the exact arithmetic
       matrix product of $A$ and $B$ running in time
       $O(n^2\min\{\ell +\lambda(A,\ell,row),k+\lambda(B,k,col)\})$.
       Alternatively, a simple randomized algorithm for the exact
       product running in time $O(n^2\log n+ \ell n k
       +n^2(\lambda(A,\ell,row)+\lambda(B,k,col))$ w.h.p.
\end{enumerate}
\junk{
  Our simple deterministic algorithms for the approximation, preprocessing,
  answering queries about single entries, and the exact arithmetic
  product of 0-1 matrices rely on Gonzalez's $2$-approximate
  algorithm for the $k$-center problem and the idea of using
  representatives of row and/or column clusters in order
  to compute a smaller matrix product and then recover the values
  of the entries in the target matrix product, already used
  for square Boolean matrices in \cite{}.}
In the fully symmetric case, i.e., when
  the matrices $A,\ B$ are square,
  $\ell=k,$ and $\lambda(A,\ell,row)=\lambda(B,k,col),$ the asymptotic
  running times of our deterministic algorithms for the arithmetic
  matrix product of 0-1 matrices $A,\ B$ coincide with the
  asymptotic running
  time of the algorithm for the Boolean matrix product presented in
  \cite{AC11}.

\subsection{Techniques}

All our main deterministic results rely on
the classical, simple $2$-approximation algorithm for
the $k$-center clustering problem (farthest-point clustering) due to
Gonzalez \cite{Gon85}, whose properties are summarized in Fact \ref{fact: gon}
below.
Our results also rely on the idea of updating
the inner product of two vectors $a$ and $b$ in $\{0,1\}^q$
over the Boolean or an arithmetic semi-ring
to that of two vectors $a'$ and $b'$ in $\{0,1\}^q$,
in time roughly
proportional to $\mathrm{ham}(a,a')+\mathrm{ham}(b,b')$.
The idea has been used
in \cite{AC11,BL,FJLL18,GL03}.
As in some of the aforementioned articles, in case
of our alternative deterministic algorithm
for 0-1 matrix multiplication, we combine it
with a traversal of an approximate minimum spanning tree of the rows or
columns of an input matrix in the Hamming space $\{0,1\}^q$,
where $q$ is the length of the rows or columns
(see Lemma \ref{lem: mmst}).

\subsection{Article organization}

The next section contains basic definitions.
Section 3 presents a
hybrid algorithm for the arithmetic matrix product of 0-1 matrices
combing clustering of the rows or columns of the input matrices with
the technique of traversing an approximate minimum spanning tree of
the rows or columns, respectively. Section 4 is devoted to our
approximation algorithms for the arithmetic product of two 0-1
matrices. Section~5 presents preprocessing procedures enabling efficient answers
to queries asking for the value of an arbitrary entry of the
arithmetic product matrix. As 
a corollary, an upper time bound on
the construction of the exact arithmetic matrix product of 0-1
matrices follows, it matches that yielded by the hybrid
algorithm from Section 3. We conclude with a short discussion on
potential extensions of our results.


\section{Preliminaries}

For a positive integer $r$, $[r]$ stands for the set of positive integers not exceeding $r.$

For two vectors $(a_1,\dots,a_d)$ and $(b_1,\dots,b_d)$ in $\mathbb{R}^d,$
their {\em inner product} is equal to $\sum_{\ell=1}^d a_{\ell}b_{\ell}.$
The transpose of a matrix $D$ is denoted by $D^{\top}.$
If the entries of $D$ are in $\{0,1\}$ then
$D$ is a 0-1 matrix.
\junk{
The symbol $\omega$ denotes the smallest real number such that two $n\times n$
matrices can be multiplied using $O(n^{\omega +\epsilon})$
operations over the field of reals, for all $\epsilon > 0.$}

The {\em Hamming distance} between two points $a,\ b$ (vectors) in $\{0,1\}^d$
is the number of coordinates in which the two points differ.
Alternatively, it can be defined as the distance between $a$ and $b$ in
the $L_1$ metric over $\{0,1\}^d.$
It is denoted by $\mathrm{ham}(a,b).$

An event is said to hold {\em with high probability} (w.h.p. for short) in
terms of a parameter $N$ related to the input size if it holds with
probability at least $1-\frac 1 {N^{\alpha}},$ for any constant
$\alpha $ not less than $1.$

The \emph{$k$-center clustering problem in a Hamming space} is as follows:
Given a set $P$ of $n$ points in a Hamming space $\{0,1\}^d$, find a set $T$ of
$k$ points in $\{0,1\}^d$ that minimize
$\max_{v\in P} \min_{u\in T} \mathrm{ham}(v,u).$
It is known to be NP-hard already for $k = 1$ \cite{FL97}
and also NP-hard to approximate within a ratio of $2-\epsilon$
for any constant $\epsilon >0$ when $k$ is part of the input \cite{FG88,GJL04}.
Note that according to the problem definition, the points in $T$ do not have to
belong to the set $P$.
The \emph{minimum-diameter $k$-clustering problem in a Hamming space} is:
Given a set $P$ of $n$ points in a Hamming space $\{0,1\}^d$, find a partition
of $P$ into $k$ subsets $P_1,P_2,\dots, P_{k}$ that minimize
$\max_{i\in [k]} \max_{v,u \in P_i} \mathrm{ham}(v,u).$
The $k$-center clustering problem can also be termed
the minimum-radius $k$-clustering problem.

\begin{fact}\cite{Gon85}\label{fact: gon}
(Gonzalez's algorithm)
   There is a simple deterministic $2$-approxi\-mation algorithm for
   the $k$-center clustering and minimum-diameter $k$-clustering problems
   in a $d$-dimensional Hamming space running in $O(ndk)$ time.
   Moreover, the approximate solution returned by Gonzalez's algorithm consists
   entirely of points from the input set~$P$.
\end{fact}


We will use the last property given in Fact~\ref{fact: gon} (that is,
the approximate solutions found by Gonzalez's algorithm consist of points from
the input only) in the proof of Theorem~\ref{theo: mmstclus}.

\medskip

A combination of a variant of randomized dimension reduction in Euclidean
spaces in the framework of Johnson and Lindenstrauss \cite{JL84} given by
Achlioptas in \cite{A03} and the observation that the Hamming distance
between two 0-1 vectors is equal to their squared $L_2$ distance makes it
possible to decrease the dimension to a logarithmic one in the number of
input points, at the cost of worsening the approximation guarantee by
some arbitrarily small $\epsilon$.
Taking into account the time needed to compute the images of the input points
in the subspace of logarithmic dimension,
this yields the following fact
shown in \cite{KLP26}.
 \begin{fact} \cite{KLP26} \label{fact: 2plus}
   For any fixed $\epsilon\in (0,1/2)$,
   there is a randomized algorithm for
   the $k$-center clustering problem
   in a Hamming space running in $O(n\log n(d + k)/\epsilon^2)$ time
   that computes a $(2+\epsilon)$-approximation
   of an optimal $k$-center clustering w.h.p.
 \end{fact}
   
 \section{A hybrid algorithm for multiplication of 0-1 matrices}

   In this section, we present a simple deterministic algorithm for the
 arithmetic matrix product of two 0-1 matrices.  It combines
 an approximate $\ell$-center clustering of the rows of the first matrix
 or an approximate $k$-center clustering of the columns of the second
 matrix with the aforementioned technique of traversing an approximate
 minimum spanning tree of the rows of the first matrix or the columns
 of the second matrix in an appropriate Hamming space in order to
 compute a row or column of the product matrix \cite{BL,FJLL18,GL03}.
 We generalize the technique to include rectangular matrices.
 The clustering is used for a fast deterministic construction of an
 approximate minimum spanning tree. The prior algorithms based on the
 traversing technique use a fast randomized construction of an
 appropriate approximate minimum spanning tree \cite{BL,FJLL18,GL03}.

 We shall use the following procedure and lemma in the spirit of
 \cite{BL,FJLL18,GL03}.

\bigskip
\noindent
\fbox{
\parbox{\boxwidth}{  
\smallskip
{\bf procedure $MMCLUS$-$ST(A,B,T)$}
\par
\smallskip
\noindent
    {\em Input:} Two matrices $A$ and $B$ of sizes $p\times q$ and $q\times r$,
    respectively, and a spanning tree $T$ of the rows of $A$ in the
Hamming space $\{0,1\}^q.$
\par
\smallskip
\noindent
    {\em Output:} The arithmetic matrix product $C$ of $A$ and $B$.
    \par
\noindent
\begin{enumerate}
\item
 Construct a traversal 
(i.e., a non-necessarily simple path visiting all vertices) $U$ of $T.$
\item For each pair composed of $m$-th row $A_{m*}$ of $A$
  and $i$-th row $A_{i*}$ of $A,$
where the latter row follows the former in
the traversal $U,$ compute the set $diff(m,i)$
of indices $h \in [q]$ where $A_{ih}\neq A_{mh}$.
\item
For $j=1,\ldots,r$ do:
\begin{enumerate}
\item Compute $C_{sj}$ where $A_{s*}$ 
is the row of $A$ from which the traversal $U$ of $T$ starts.
\item While following $U$, do:
\begin{enumerate}
\item
Set $m,\ i$ to the indices
of the previously traversed row of $A$ and
the currently traversed row of $A$, respectively.
\item
Set $C_{ij}$ to $C_{mj}$.
\item For each $h \in diff(m,i)$,
  if $A_{ih}B_{h j}=1$ then set
  $C_{ij}$ to $C_{ij}+1$ and if  $A_{mh}B_{hj}=1$ then set
  $C_{ij}$ to $C_{ij}-1$.
\end{enumerate} 
\end{enumerate} 
\end{enumerate}
}}
\bigskip

Define the Hamming cost $\mathrm{ham}(S)$ of a spanning tree $S$
of a point set $P \subset \{0,1\}^d$ by
$\mathrm{ham}(S)=\sum_{(v,u)\in S} \mathrm{ham}(v,u).$

\begin{lemma}\label{lem: mmst}
Let $A$ and $B$ be two 0-1 matrices of sizes $p\times q$ and
 $q\times r,$ respectively. Given a spanning tree $T_A$ of the rows
 of $A$  in the Hamming space $\{0,1\}^q$,
   $MMCLUS$-$ST(A,B,T_A)$ computes
  the arithmetic matrix product of $A$ and $B$
in time $O(pq+qr+pr+r\times \mathrm{ham}(T_A)).$ 
\end{lemma} 
\begin{proof}
  The correctness of $MMCLUS$-$ST(A,B,T_A)$
  follows from the correctness of
  the updates of $C_{ij}$ in the block of the inner loop, i.e.,
  in Step~3(b).
  Step 1 of
  $MMCLUS$-$ST(A,B,T_A)$
  can be done in $O(p)$ time 
while Step 2 requires
  $O(pq)$ time. The first step in the block under
  the outer loop, i.e., computing
  $C_{sj}$ in Step 3(a), takes $O(q)$ time.
  The crucial observation is that the second step
  in this block, i.e., Step 3(b), requires $O(p+\mathrm{ham}(T_A))$ time.
  Simply, the substeps (i), (ii) take $O(1)$ time while
  the substep (iii) requires $O(|diff(m,i)|+1)$ time.
  Since the block is iterated $r$ times, the whole
  outer loop, i.e., Step 3,
  requires $O(qr+pr+r\mathrm{ham}(T_A))$ time.
  Thus,
  $MMCLUS$-$ST(A,B,T_A)$
  can be implemented in time
  $O(pq+qr+rp+r\times \mathrm{ham}(T_A))$.
\junk{
  Similarly, we can run
  $MMCLUS$-$ST(B^{\top},A^{\top},T_B)$
  to obtain the transpose of the arithmetic matrix
  product of $A$ and $B.$  So, to obtain the lemma,
  we can compute $\mathrm{ham}(T_A)$ and $\mathrm{ham}(T_B)$
  in $O(pq+qr)$ time in order to run  whichever 
  of $MMCLUS$-$ST(A,B,T_A)$ and $MMCLUS$-$ST(B^{\top},A^{\top},T_B)$
  has a smaller upper time bound.}
\qed
\end{proof}

\begin{figure}[t]
 \begin{center}
   \includegraphics[width=0.9\textwidth]{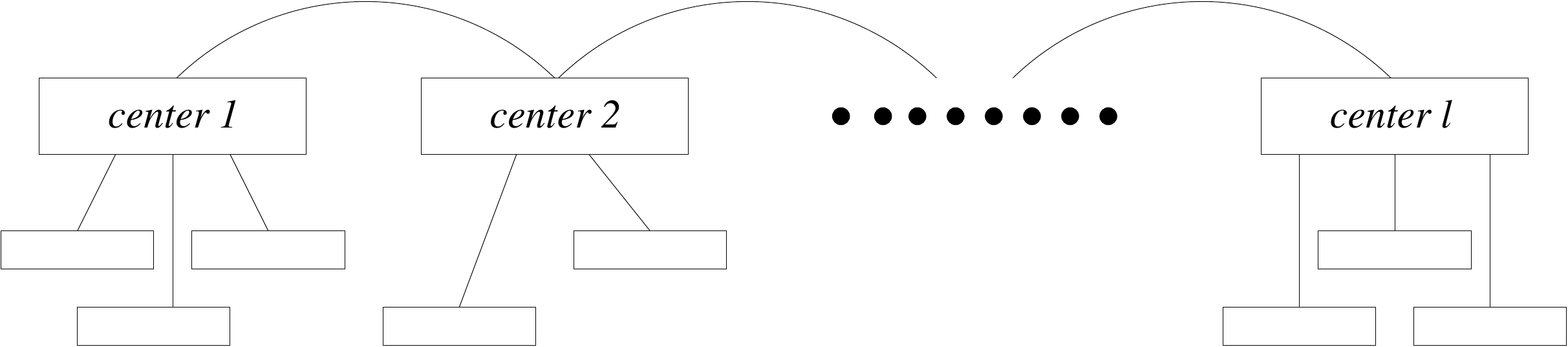}
\end{center}
\caption{An example of the spanning tree $T_A$ of the rows of $A$.}
\label{fig: sp}
\end{figure}

\begin{theorem}\label{theo: mmstclus}
  Let $A$ and $B$ be two 0-1 matrices of sizes $p\times q$ and
  $q\times r,$ respectively. Given parameters $\ell \in [p]$
and $k\in [r],$ the arithmetic matrix product of $A$ and $B$
can be computed by a simple deterministic algorithm in time
$O(pr+\min\{pq\ell+rq\ell +pr\times \lambda(A,\ell,row),
rqk+pqk+pr\times \lambda(B,k,col)\}).$
\end{theorem}
\begin{proof}
We can determine an $\ell$-center clustering of the rows of
$A$ in $\{0,1\}^q$ of maximum cluster radius not exceeding
$2\lambda(A,\ell,row)$ in $O(pq\ell)$ time by employing
Fact \ref{fact: gon}. Similarly, we
can construct a $k$-center clustering of the columns of $B$
in $\{0,1\}^q$ of maximum cluster radius not exceeding
$2\lambda(B,k,col)$ in $O(rqk)$ time.
The centers in both aforementioned clusterings are some rows
of $A$ and some columns of $B$, respectively, by
the last part of Fact~\ref{fact: gon}.
Hence, the $\ell$-center clustering gives rise to a spanning tree $T_A$
of the rows of $A$ with all members of a cluster being pendants
of their cluster center and the centers connected by
a spanning tree of the centers, e.g., a path of length
$\ell-1,$ see Fig. \ref{fig: sp}.
The Hamming cost of $T_A$ is at most $(p-\ell)2\lambda(A,\ell,row)+
(\ell -1)q$.
  \footnote{In case the minimum Hamming cost of a spanning tree
  of the centers was substantially smaller than $(\ell -1)q$, one could
  reduce the term  $rq\ell$ in the upper bound, at the cost of
  increasing the running time of the algorithm
  by that needed for the construction of a nearly optimal
spanning tree of the centers.}
Similarly, we can obtain a spanning tree $T_B$ of
the columns of $B$ having the Hamming cost not exceeding
$(r-k)2\lambda(B,k,col) + (k-1)q.$
Note that the arithmetic matrix product of $B^{\top}$ and $A^{\top}$
is equal to the transpose of the arithmetic matrix product of
$A$ and $B.$  Hence,
to obtain the theorem, we can alternate the steps of
$MMCLUS$-$ST(A,B,T_A)$ with the steps of
$MMCLUS$-$ST(B^{\top},A^{\top},T_B)$ and stop whenever
any of the calls is completed. Note also that a $k$-clustering of
columns of $B$ is equivalent to
a $k$-clustering of the rows
of $B^{\top}$ and hence $\lambda(B^{\top},k,row)=\lambda(B,k,col).$
The theorem
follows from Lemma \ref{lem: mmst} by
substituting the upper bounds on the Hamming cost
of $T_A$ and $T_B$ for $\mathrm{ham}(T_A)$ and $\mathrm{ham}(T_B)$
and adding the asymptotic running time of
Gonzalez's algorithm for $\ell$-center and $k$-center
clusterings, respectively.
\qed
\end{proof}

 \section{An approximate arithmetic matrix product of 0-1 matrices}

 In this section we present two algorithms for approximate arithmetic
 matrix product of 0-1 matrices. The first is deterministic while the
 second is randomized. The general idea is to find $\ell$
 approximate centers of the rows of the first matrix or/and
 $k$ approximate centers of the columns of the second matrix
 and then use the centers instead of the rows or/and the columns
 while computing a matrix product. The latter product is a base
 for the approximate matrix product. A similar idea has been used
 in \cite{AC11} to compute the exact Boolean matrix product
 of square Boolean matrices.
 \junk{
 A similar idea has been used
 in \cite{AC11} to compute the Boolean matrix product
 of two Boolean matrices.
 The algorithms are the base of preprocessing
 enabling efficient answers to queries asking for the exact value of a
 single entry in the matrix product. The preprocessings and simple
 algorithms for answering the queries are presented in
 the next section.}
 \junk{
 The idea of  the procedure is very simple.
 First, it runs an approximate algorithm for the $\ell$-center
 problem on the rows of $A.$ 
 Next, it forms a smaller matrix $A'$
 whose $\ell$ rows are the computed centers
 of the rows of $A$. Then, it computes the arithmetic matrix product
 $D$ of the matrices $A'$ and $B.$}


Our deterministic
approximation algorithm for the arithmetic matrix product of two
0-1 matrices is specified by the following procedure.

\bigskip
\noindent
\fbox{
\parbox{\boxwidth}{  
\smallskip
        {\bf procedure} $MMCLUS$-$Approx(A,B, \ell)$
        \par
\smallskip
        \noindent
            {\em Input}: Two 0-1 matrices $A$ and $B$  of sizes $p\times q$ and $q\times r,$
            respectively,
            and a positive integer $\ell$
            not exceeding $p$.
            \par
\smallskip
            \noindent
                {\em Output}: A $p\times r$ matrix $D$, where
                for $1\le i\le p$ and $1\le j\le r,$ $D_{ij}$
                is an approximation of the inner product $C_{ij}$ of
                the $i$-th row $A_{i*}$ of $A$  and the $j$-th column
                $B_{*j}$ of $B.$
                \begin{enumerate}
                \item
                  \label{step: Center}
                                    Determine an approximate solution
                                    to the $\ell$-center clustering
                  of the rows of the matrix $A$ in $\{ 0,1\}^q.$
                  For each row $A_{i*}$ of $A$, set
                  $cen_{\ell}(A_{i*})$ to its closest center,
                  with ties broken arbitrarily.
                \item
                  \label{step: Form_A'}
                  Form the $\ell \times q$ matrix $A',$ where
                  the $i'$-th row is the $i'$-th center in
                  the approximate $\ell$-center clustering
                  of the rows of $A.$
                 \item
                   \label{step: Compute_C'}
                   Compute the arithmetic $\ell \times r$ matrix  product
                   $C'$ of $A'$ and $B.$ 
                 \item
                   \label{step: Compute_D}
                   For $1\le i\le p$ and $1\le j \le r,$ set
                   $D_{ij}$ to $C'_{i'j}$, where the $i'$-th row
                   $A'_{i'*}$ of $A'$ is $cen_{\ell}(A_{i*}).$
                \end{enumerate}
}}
\bigskip

The idea of  the procedure is very simple.
 First, it runs an approximate algorithm for the $\ell$-center
 problem on the rows of $A$ in step~\ref{step: Center}. 
 Next, it forms a smaller matrix $A'$
 whose $\ell$ rows are the computed centers
 of the rows of $A$ in step~\ref{step: Form_A'}.
 Then, it computes the arithmetic matrix product
 $C'$ of the matrices $A'$ and $B$ in step~\ref{step: Compute_C'},              see Fig.~\ref{fig: fig_matrix_multiplication}.
%
After having computed~$C'$, the procedure sets each entry~$D_{ij}$ of
the approximate matrix~$D$ to the entry of~$C'$ holding the inner product of
the center closest to the $i$-th row with the $j$-th column of~$B$ in
step~\ref{step: Compute_D}.

\begin{figure}[t]
\centering
\includegraphics[width=0.9\textwidth]{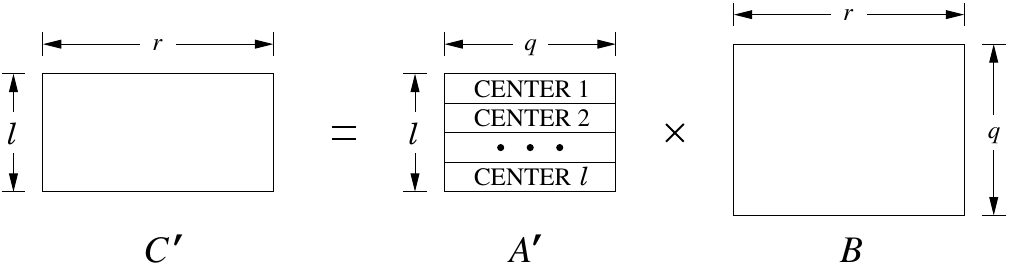}
\medskip
\caption{Computing the matrix~$C'$ in step~\ref{step: Compute_C'} of
$MMCLUS$-$Approx(A,B, \ell)$.}
\label{fig: fig_matrix_multiplication}
\end{figure}

                 For a 0-1 $p\times q$ matrix $A$, let
                 $\lambda(A,\ell,row)$ be the minimum, over all
               $\ell$-center clusterings of the rows of $A$ in the
                 Hamming space $\{0,1\}^q$, of the maximum Hamming
                 distance between a row of $A$ and the closest center.
                 Similarly, for a 0-1 $q\times r$ matrix
                 $B,$ let $\lambda(B,k,col)$ be the minimum, over all
                 $k$-center clusterings of the columns of $B$ in the
                 Hamming space $\{0,1\}^q$, of the maximum Hamming
                 distance between a column of $B$
                 and the closest center.
                 
                  \begin{lemma}\label{lem: appapp}
                   Suppose a $2$-approximation algorithm for
                   the $\ell$-center clustering
                   is used in $MMCLUS$-$Approx(A,B, \ell)$
                   and let $C$ stand for the arithmetic product of $A$
                   and $B.$ Then, for  $1\le i\le p$ and $1\le j \le r,$
                   $|C_{ij}-D_{ij}|\le 2\lambda(A,\ell,row)$.
                   
 \end{lemma}
 \begin{proof} 
   For  $1\le i\le p$ and $1\le j \le r,$  $D_{ij}$ is the inner product
   of $cen_{\ell}(A_{i*})$, where $\mathrm{ham}(A_{i*},cen_{\ell}(A_{i*}))\le
   2\lambda(A,\ell,row)$, with $B_{*j}.$
 Hence, $C_{ij}$, which is
   the inner product of $A_{i*}$ with $B_{*j},$ can differ at most by
   $2\lambda(A,\ell,row)$ from $D_{ij}.$
 \end{proof}

 By $T(s,q,t)$, we shall denote the worst-case time taken by the
multiplication of two 0-1 matrices of sizes $s \times q$ and
 $q\times t,$ respectively.
 
 \begin{lemma}\label{lem: apptime}
   $MMCLUS$-$Approx(A,B, \ell)$
   runs in $O(pq\ell +pr+T(\ell,q,r))$ time.
     \end{lemma}
  \begin{proof}
   Step 1, which includes the assignment of the closest center
   to each row of $A$, can be done in $O(pq\ell)$ time by using Fact
   \ref{fact: gon}.
   Step 2 takes $O(\ell q)$ time, which is $O(T(\ell,q,r))$ time.
   Finally, Step 3 takes $T(\ell,q,r)$ time while Step 4
   can be done in $O(pr)$ time. Thus, the overall time is
   $O(pq\ell +pr +T(\ell,q,r)).$
 \end{proof}

 We can use the straightforward $O(sqt)$-time algorithm for the
 multiplication of two matrices of sizes $s\times q$ and
 $q\times t,$ respectively.
 Since $T(\ell,q,r) = O(\ell q r)$,
 Lemmata \ref{lem: appapp} and
\ref{lem: apptime} yield
the first variant (1) of our first main
result (Theorem~\ref{theo: app} below).
The second variant (2) 
follows from the first one by $(AB)^{\top}=B^{\top}A^{\top}.$
Simply, we run
$MMCLUS$-$Approx(B^{\top},A^{\top},k)$
to compute an approximation of the
transpose of the arithmetic matrix product of
$A$ and $B.$ Note also that a $k$-clustering of
columns of $B$ is equivalent to
a $k$-clustering of the rows
of $B^{\top}$ and hence $\lambda(B^{\top},k,row)=\lambda(B,k,col).$

\begin{theorem}\label{theo: app}
Let
$A$ and $B$  be two 0-1 matrices
of sizes $p\times q$ and $q\times r,$
respectively. There is a simple deterministic
algorithm which provides an approximation
of all entries of the arithmetic matrix product of $A$ and $B$
within an additive error of at most:
\begin{enumerate}
\item
$2\lambda(A,\ell,row)$ in time $O(pq\ell +\ell qr+ pr),$ or
\item
$2\lambda(B,k,col)$ in time $O(rqk+pqk+pr)$, respectively. 
\end{enumerate}
\end{theorem}

The heavy three-linear terms $pq\ell,$ $\ell qr,$
$rqk,$ and $pqk$ in the upper time bounds in Theorem~\ref{theo: app}
originate from the upper time bound on the center
clustering problem in Fact \ref{fact: gon} and the time taken by
matrix multiplication. We can get rid of these terms
by using a faster randomized algorithm for approximate
center clustering according to Fact \ref{fact: 2plus}, both on the rows
of the first matrix $A$  and the columns of the second matrix $B$,
in a way similar to \cite{AC11}.
Then, only the small matrices induced by the centers of
the rows of $A$ and the columns of $B$, respectively, need to be multiplied.
In this way, the terms $pq\ell$ and $rqk$ are replaced
by the presumably smaller term $\ell qk$ in the upper time
bound at the cost of increasing the additive error.

\bigskip
\noindent
\fbox{
\parbox{\boxwidth}{  
\smallskip
        {\bf procedure} $MMCLUS$-$R$-$Approx(A,B, \ell,k, \epsilon)$
        \par
\smallskip
        \noindent
            {\em Input}: Two 0-1 matrices $A$ and $B$  of sizes $p\times q$ and $q\times r,$
            respectively,
            positive integers $\ell,\ k$
            not exceeding $p,\ r,$ respectively,
            and a real number $\epsilon\in (0,1/2).$
            \par
\smallskip
            \noindent
                {\em Output}: A $p\times r$ matrix $D'$, where
                for $1\le i\le p$ and $1\le j\le r,$ $D'_{ij}$
                is a $(2+\epsilon)$-approximation of the inner product $C_{ij}$ of
                the $i$-th row $A_{i*}$ of $A$  and the $j$-th column
                $B_{*j}$ of $B.$
                \begin{enumerate}
                                  \item
                  Determine an approximate $\ell$-center clustering
                  of the rows of the matrix $A$ in $\{ 0,1\}^q$
                  by using Fact \ref{fact: 2plus}.
                  For each row $A_{i*}$ of $A$, set
                  $cen_{\ell}(A_{i*})$ to its closest center,
                  with ties broken arbitrarily.
                  \item
                  Determine an approximate $k$-center clustering
                  of the columns of the matrix $B$ in $\{ 0,1\}^q$
                  by using Fact \ref{fact: 2plus}.
                  For each column $B_{*j}$ of $B$, set
                  $cen_{k}(B_{*j})$ to its closest center,
                  with ties broken arbitrarily.
                \item
                  Form the $\ell \times q$ matrix $A',$ where
                  the $i'$-th row is the $i'$-th center in
                  the approximate $\ell$-center clustering
                  of the rows of $A.$
                  \item
                  Form the $q\times k$ matrix $B',$ where
                  the $j'$-th column is the $j'$-th center in
                  the approximate $k$-center clustering
                  of the columns of $B.$
                 \item
                   Compute the arithmetic $\ell \times k$ matrix  product
                   $C''$ of $A'$ and $B'.$ 
                 \item
                   For $1\le i\le p$ and $1\le j \le r,$ set
                   $D'_{ij}$ to $C''_{i'j'}$, where the $i'$-th row
                   $A'_{i'*}$ of $A'$ is $cen_{\ell}(A_{i*})$ and
                   the $j'$-th column $B'_{*j'}$ of $B'$ is
                   $cen_{k}(B_{*j})$.
                 \end{enumerate}
}}
\bigskip

 \begin{lemma}\label{lem: rappapp}
   Let $C$ stand for the arithmetic product of $A$
                   and $B.$ For  $1\le i\le p$ and $1\le j \le r,$
                   $|C_{ij}-D'_{ij}|\le (2+\epsilon)(\lambda(A,\ell,row)
                   +\lambda(B,k,col))$ holds w.h.p.
                   
 \end{lemma}
 \begin{proof} 
   For  $1\le i\le p$ and $1\le j \le r,$  $D'_{ij}$ is the inner product
   of $cen_{\ell}(A_{i*})$, where $\mathrm{ham}(A_{i*},cen_{\ell}(A_{i*}))\le
   (2+\epsilon)\lambda(A,\ell,row)$ w.h.p.,
   with $cen_{k}(B_{*j})$,
   where $\mathrm{ham}(B_{*j},cen_{k}(B_{*j}))\le
(2+\epsilon)\lambda(B,k,col)$ w.h.p. by Fact \ref{fact: 2plus}.
 Hence, $C_{ij}$, which is
   the inner product of $A_{i*}$ with $B_{*j},$ can differ at most by
   $(2+\epsilon)(\lambda(A,\ell,row)+\lambda(B,k,col))$ from $D'_{ij}$
   w.h.p.
 \end{proof}
                 
 \begin{lemma}\label{lem: rapptime}
   $MMCLUS$-$R$-$Approx(A,B, \ell,k,\epsilon)$
   can be implemented in time \linebreak
   $O(p\log p(q + \ell)/\epsilon^2 +r\log r(q + k)/\epsilon^2+pr+T(\ell,q,k))$.
     \end{lemma}
  \begin{proof}
      In Step 1, the $\ell$ centers of the rows of $A$
    can be found in $O(p\log p(q + \ell)/\epsilon^2)$ time
    by Fact \ref{fact: 2plus}. Also, the $\ell$ clusters induced
    by the centers can be formed in $O(p\log p(q + \ell)/\epsilon^2)$ time
    in order to approximate minimum-diameter $\ell$-clustering
    of the rows of $A$; see the proof of Fact \ref{fact: 2plus} in \cite{KLP26}.
    Importantly, each member of such a cluster is within a Hamming distance
    not exceeding $(2+\epsilon)\lambda(A, \ell, row)$ from its center.
    Symmetrically, Step 2 takes $O(r\log r(q + k)/\epsilon^2)$ time.
    Steps  3 and 4 can be done 
    $O(\ell q)$ time and $O(kq)$ time, respectively,
    which is $O(T(\ell,q,k))$ time.
   Finally, Step 5 takes $T(\ell,q,k)$ time while Step 6
   can be done in $O(pr)$ time. Thus, the overall time is
   $O(p\log p(q + \ell)/\epsilon^2 +r\log r(q + k)/\epsilon^2+pr+T(\ell,q,k))$.
  \end{proof}

  Lemmata \ref{lem: rappapp} and
\ref{lem: rapptime} yield our second theorem.
                 
\begin{theorem}\label{theo: rapp}
Let
$A$ and $B$  be two 0-1 matrices
of sizes $p\times q$ and $q\times r,$
respectively. There is a simple randomized
algorithm which for any $\epsilon\in (0,1/2)$ provides an approximation
of all entries of the arithmetic matrix product of $A$ and $B$
within an additive error of at most
  $(2+\epsilon)(\lambda(A,\ell,row)+\lambda(B,k,col))$ w.h.p. in time
  $O(p\log p(q + \ell)/\epsilon^2+r\log r(q + k)/\epsilon^2+
  \ell q k +pr)$.
\end{theorem}

\section{Preprocessing for answering single-entry queries}
In this section, we use our algorithms for an approximate matrix product
of 0-1 matrices for the construction of a preprocessing of the input
matrices enabling fast answers to queries asking for the exact
value of a single entry in the product matrix. As corollaries,
we obtain upper time bounds on computing the exact matrix product,
one of them matching that of Theorem~\ref{theo: mmstclus}.

We can apply
$MMCLUS$-$Approx(A,B, \ell)$
to obtain a preprocessing for answering queries about single entries of
the arithmetic matrix product of $A$ and $B.$
In the first step of the preprocessing,
$MMCLUS$-$Approx(A,B, \ell)$ is run. In the second step,
for each row of $A$, the set of indices of coordinates
on which it differs from its closest center is computed.

\bigskip
\noindent
\fbox{
\parbox{\boxwidth}{  
\smallskip
        {\bf procedure} $MMCLUS$-$Preproc(A,B, \ell)$
        \par
\smallskip
        \noindent
            {\em Input}: Two 0-1 matrices $A$ and $B$  of sizes $p\times q$ and $q\times r,$
            respectively,
            and a positive integer $\ell$ 
            not exceeding $p$.
            \par
\smallskip
            \noindent
            {\em Output}: The $p\times r$ matrix $D$
            returned by $MMCLUS$-$Approx(A,B, \ell)$,
and for $1\le i\le p$, the set 
of coordinate indices $ind(A,i)$ on
which $A_{i*}$  differs from the closest center.
 
\begin{enumerate}
                 \item
                   Run $MMCLUS$-$Approx(A,B, \ell)$.
                 \item
                  
 For $1\le i \le p,$ determine the set $ind(A,i)$
                   of coordinate indices on which $A_{i*}$ differs
                   from $cen_{\ell}(A_{i*})$.
                    \end{enumerate}
}}
\bigskip

   \begin{lemma}\label{lem: pretime}
   $MMCLUS$-$Preproc(A,B, \ell)$
rums in
   $O(pq\ell+pr
   +T(\ell,q,r))$ time.
    \end{lemma}
 \begin{proof} 
   Step 1 can be done in $O(pq\ell +pr+T(\ell,q,r))$ time
by Lemma \ref{lem: apptime}. Step 2 
   can easily be implemented in $O(pq)$ time.
\qed
 \end{proof}

\medskip

The idea of our procedure for answering a query about a single entry $C_{ij}$
of the matrix product of $A$ and $B$ is to recover
the exact value of the entry from its approximation. 
Recall that the approximation is the inner
product of the center assigned to the $i$-th row
of $A$ with the $j$-th column of $B$. The correction
takes time proportional to the size of the set of
indices of coordinates on which the $i$-th row and
the center differ. See the procedure $MMCLUS$-$Query(A,B, \ell,i,j)$
for details.

\bigskip
\noindent
\fbox{
\parbox{\boxwidth}{  
\smallskip
        {\bf procedure} $MMCLUS$-$Query(A,B, \ell,i,j)$
        \par
\smallskip
        \noindent
            {\em Input}: The preprocessing done by
            $MMCLUS$-$Preproc(A,B,\ell)$
            for 0-1 matrices $A$ and $B$  of sizes $p\times q$ and $q\times r,$
respectively,
$\ell\in [p]$,  
            and two query indices $i\in [p]$ and $j\in [r].$
            \par
\smallskip
            \noindent
                {\em Output}: The inner product $C_{ij}$ of
                the $i$-th row $A_{i*}$ of  $A$  and the $j$-th column
                $B_{*j}$ of $B.$
                \begin{enumerate}
                
                 \item
                   Set $C_{ij}$ to the entry $D_{ij}$  of the
                   matrix $D$ computed by
                   $MMCLUS$-$Approx(A,B, \ell)$ in
                   $MMCLUS$-$Preproc(A,B,\ell)$.
\item
                   For $m\in ind(A,i)$ do
                   \begin{enumerate}
                 \item 
                     If the $m$-th coordinate
                     of the center assigned to $A_{i*}$ is $0$
                     and $B_{mj}=1$ then $C_{ij}\leftarrow
                     C_{ij}+1$.
                   \item
                     If the $m$-th coordinate of the center
                     assigned to $A_{i*}$ is $1$ and $B_{mj}$ is also $1$
                  then $C_{ij}\leftarrow
                  C_{ij}-1$.
 \end{enumerate}
\end{enumerate}
}}
\bigskip

\begin{lemma} \label{lem: querycor}
$MMCLUS$-$Query(A,B, \ell,i,j)$
is correct, i.e., the final value of $C_{ij}$ is the inner product
of the $i$-th row $A_{i*}$ of  $A$  and the $j$-th column
$B_{*j}$ of $B.$
\end{lemma}
              \begin{proof}
                $C_{ij}$ is initially set to
                $D_{ij}$, which is the inner product of the center
                assigned to $A_{i*}$ and $B_{*j}.$ Then, $C_{ij}$ is
                appropriately corrected by increasing or decreasing
                with $1$ for each coordinate index $m\in ind(A,i)$
                which contributes $1$ to the inner product of $A_{i*}$
                and $B_{*j}$ and $0$ to the inner product of the
                center of $A_{i*}$ and $B_{*j}$ or {\em vice versa}.
              \end{proof}

              \begin{lemma}\label{lem: querytime}
                $MMCLUS$-$Query(A,B,\ell,i,j)$
                takes $O(\lambda(A,\ell,row)+1)$ time.
              \end{lemma}
             \begin{proof}
               Recall that $2\lambda(A,\ell,row)$
               is an upper bound on the maximum
               Hamming distance between a row of $A$ and the closest center
               in the $\ell$-center clustering computed by
               $MMCLUS$-$Approx(A,B, \ell)$ in $MMCLUS$-$Preproc(A,B,\ell)$.
Step 1 takes $O(1)$ time.
               Since
              the $m$-th coordinate in the centers can be accessed
               in the matrix $A'$ computed by
               $MMCLUS$-$Approx(A,B, \ell)$, 
               each of the two substeps
               in the block of the loop in Step 2 can be done in
               $O(1)$ time.  Finally, since
               $|ind(A,i)|\le 2\lambda(A,\ell,row)$,
               the block is iterated at most
               $2\lambda(A,\ell,row)$ times.
               Consequently, the whole Step 2 takes
               $O(\lambda(A,\ell,row)+1)$ time.
             \end{proof}

             By putting Lemmata \ref{lem: pretime}, \ref{lem:
               querycor}, and \ref{lem: querytime} together,
             and using the
             straightforward $O(sqt)$-time algorithm to multiply
             matrices of size $s\times q$ and $q\times t,$ we
             obtain the first variant of
             our next main result.
             The second variant reduces to the first one
             by $(AB)^{\top}=B^{\top}A^{\top}.$
             We simply run
             $MMCLUS$-$Preproc(B^{\top},A^{\top},k)$
             and
             $MMCLUS$-$Query(B^{\top},A^{\top},k ,j,i)$
             instead.

\begin{theorem}\label{theo: query}
  Let $A$ and $B$ be two 0-1 matrices of sizes $p\times q$ and
  $q\times r,$ respectively. Given parameters $\ell \in [p]$
  and $k\in [r],$ the matrices can be preprocessed by a simple
deterministic algorithm
in $O(pq\ell +\ell qr+pr)$ time or 
$O(rqk + pqk+pr)$ time 
such that a query asking for the exact  value of a single entry
$C_{ij}$ of the arithmetic matrix product $C$
of $A$ and $B$ can be answered in 
$O(\lambda(A,\ell,row)+1)$ time 
or $O(\lambda(B,k,col)+1)$ time, respectively.
\end{theorem}

Analogously, by using Lemmata \ref{lem: rappapp}
  and \ref{lem: rapptime},
  we obtain an alternative, randomized variant of
  Theorem \ref{theo: query}.

  To begin with, we modify
  $MMCLUS$-$Preproc(A,B,\ell)$
  to $MMCLUS$-$R$-$Preproc(A,B,\ell,k, \epsilon).$
  The modified procedure runs
  $MMCLUS$-$R$- \linebreak $Approx(A,B, \ell,k,\epsilon)$
  instead of $MMCLUS$-$Approx(A,B, \ell)$ in its first step.
  Its second step is the same as in $MMCLUS$-$Preproc(A,B,\ell)$.
  In an additional third step,  the set $ind(B,j)$
 of coordinate indices on which $B_{*j}$ differs
 from $cen_{k}(B_{*j})$ is computed for $1\le j \le r.$
 By Lemma \ref{lem: rapptime}
 and the fact that the second and third
 steps can be done in $O(pq+rq)$ tine, we infer that
 $MMCLUS$-$R$-$Preproc(A,B,\ell,k, \epsilon)$
 can be implemented in time
 $O(p\log p(q + \ell)/\epsilon^2 +r
 \log r(q + k)/\epsilon^2+pr+T(\ell,q,k))$.
 
 Next, we modify $MMCLUS$-$Query(A,B,\ell,i,j)$
 to $MMCLUS$-$R$- \linebreak $Query(A,B,\ell,k,\epsilon,i,j),$
 replacing $MMCLUS$-$Preproc(A,B,\ell)$ with
 $MMCLUS$-$R$-$Preproc(A,B,\ell,k, \epsilon)$ in the input.
 In the first analogous step of the modified procedure,
 $C_{ij}$ is set to the entry $D'_{ij}$  of the
  matrix $D'$ computed by the call of
$MMCLUS$-$R$-$Approx(A,B, \ell,k,\epsilon)$ in
$MMCLUS$-$R$-$Preproc(A,B,\ell,k,\epsilon)$.
The second step is the same as in $MMCLUS$-$Query(A,B,\ell,i,j)$.
In an additional third step, $C_{ij}$ is symmetrically updated
with respect to $ind(B,j)$ as follows:
\par
\noindent
For $m\in ind(B,j),$
if the $m$-th coordinate
 of the center assigned to $B_{*j}$ is $0$
  and $A_{im}=1$ then $C_{ij}$ is increased by $1.$
 Also, if the $m$-th coordinate of the center
  assigned to $B_{*j}$ is $1$ and $A_{im}$ is also $1$
  then $C_{ij}$ decreased by $1.$
  
  The correctness of $MMCLUS$-$R$-$Query(A,B,\ell,k,\epsilon,i,j)$ is
  obvious as \linebreak that of $MMCLUS$-$Query(A,B,\ell,i,j)$. By Lemma
  \ref{lem: rappapp},
  $MMCLUS$-$R$- \linebreak $Query(A,B,\ell,k,\epsilon,i,j)$
  takes $O(\lambda(A,\ell,row)+\lambda(B,k,col )+1)$ time w.h.p.

  By the correctness of $MMCLUS$-$R$-$Query(A,B,\ell,k,\epsilon,i,j)$ and
  $MMCLUS$-$R$-$Preproc(A,B,\ell,k,\epsilon),$ and the upper time
  bounds on the running times of these procedures, we obtain the
  following theorem.  
 
\begin{theorem}\label{theo: rquery}
  Let $A$ and $B$ be two 0-1 matrices of sizes $p\times q$ and
  $q\times r,$ respectively. Given parameters $\ell \in [p]$ and
  $k\in [r],$
  the matrices can be preprocessed by a simple
randomized  algorithm
in time  $O(p\log p(q + \ell)+r\log r(q + k)+
  \ell q k +pr)$
such that a query asking for the exact  value of a single entry
$C_{ij}$ of the arithmetic matrix product $C$
of $A$ and $B$ can be answered in 
$O(\lambda(A,\ell,row)+\lambda(B,k,col )+1)$ time w.h.p.
\end{theorem}

Theorems \ref{theo: query} and \ref{theo: rquery}
  can be especially useful in situations where some moderate number
  of entries in the arithmetic product matrix of
  the input 0-1 matrices are of interest.
  Then, the overall time taken by the preprocessing
  and the queries might be substantially smaller than that
  when computing each of the entries separately as well as that
  needed to compute the whole matrix product.

\junk{
In this section, we present three upper time bounds on computing the
exact matrix product of two 0-1 matrices.  The first two directly
follow from our time bounds on preprocessing enabling an efficient
reporting of the exact value of a single entry in the matrix product,
established in the preceding section. Simply, after the
preprocessing, it is sufficient to query about the exact value of each
entry in the product. The third upper time bound is based on a hybrid
approach combing clustering of the rows or columns of the input
matrices with the idea of traversing a spanning tree of the rows or
columns of relatively low Hamming cost in order to compute the product.}

Theorem \ref{theo: query} yields the following corollary.

\begin{corollary}\label{cor: query}
  Let $A$ and $B$ be two 0-1 matrices of sizes $p\times q$ and
  $q\times r,$ respectively. Given parameters $\ell \in [p]$
and $k\in [r],$ the arithmetic matrix product of $A$ and $B$
can be computed by a simple deterministic
algorithm in time
  $O(pr+\min\{pq\ell + \ell qr+pr\lambda(A,\ell,row),
  rqk + pqk+ pr\lambda(B,k,col)\})$.   
\end{corollary}
\begin{proof}
 It follows from Theorem \ref{theo: query} that by
  querying for the value of each entry in the
  arithmetic matrix product $A$ and $B$
  we can compute the product in time
  \begin{enumerate}
  \item
    $O(pq\ell + \ell qr+pr(\lambda(A,\ell,row)+1))$ or
  \item
    $O(rqk + pqk+ pr(\lambda(B,k,col)+1)).$
  \end{enumerate}
To obtain the upper time bound stated
in the theorem, we can just alternate the
  steps of the algorithm yielding the first upper bound
  with the steps of the algorithm yielding the second upper bound
  and stop whenever any of these algorithms stops.
\junk{
  To obtain the third upper bound, we run Gonzalez's algorithm twice, first
  on the rows of the matrix $A$ and then on the columns of the matrix $B$
  in order to determine $2$-approximations of $\lambda(A,\ell,row)$ and
  $\lambda(B,k,col)$, respectively. Knowing the approximations,
  we can choose to continue
  with the variant incurring a better upper time bound.}
\qed
\end{proof}
  
Similarly, we obtain the following corollary
from Theorem \ref{theo: rquery}.

\begin{corollary}\label{cor: rquery}
  Let $A$ and $B$ be two 0-1 matrices of sizes $p\times q$ and
  $q\times r,$ respectively. Given parameters $\ell \in [p]$
  and $ k\in [r],$
  the arithmetic matrix product of $A$ and $B$
can be computed by a simple randomized 
algorithm in time 
$O(p\log p(q + \ell)+r\log r(q+ k)+
  \ell q k  +pr(\lambda(A,\ell,row)+\lambda(B,k,col)+1))$ w.h.p.
\end{corollary}

Note that the upper time bound from
  Corollary~\ref{cor: query}
  coincides with
  that from Theorem \ref{theo: mmstclus}.
  However,
  the origins of the terms $\ell qr$ and $pqk$ in these
  two upper bounds
  are quite different. In case of Corollary~\ref{cor: query}, the
  terms  are caused by  the use of the straightforward
  matrix multiplication algorithm. 
They clearly could be improved if we used fast arithmetic
  multiplication, but then we would lose the simplicity of our
  algorithm. In case of Theorem~\ref{cor: query}, these terms reflect
  the asymptotic worst-case Hamming cost of connecting the $\ell$ or
  $k$ centers by an easily construable spanning tree. The cost
  might be
  much lower but it depends on the centers\footnote{
    One could also connect the centers to their $1$-median
    whose $i$-th coordinate is $1$ if the majority
    of centers have $1$ on this coordinate otherwise
    it is $0$. In this way the total Hamming cost of connecting
    the centers reduces 
    to $\ell q/2$ or $kq/2,$ respectively,
    but this does not yield any asymptotic improvement.}.
Finally, in the fully symmetric case, where
$p=q=r=n,$ $\ell=k,$ and $\lambda(A,\ell,row)=\lambda(B,k,col),$
  the upper time bound from Corollary~\ref{cor: query} and
  Theorem~\ref{theo: mmstclus} coincides with that for
  the Boolean matrix product established in \cite{AC11}.
\junk{
For example, suppose that $p=q=r=n$, $\ell=O(\sqrt n),$
$k=O(\sqrt n)$, and $\lambda(A,\ell,row)=O(\sqrt n)$,
$\lambda(B,k,col))=O(\sqrt n)$.
Then, both Corollary \ref{cor: query} and
Theorem \ref{theo: mmstclus} yield  simple deterministic 
algorithms for the arithmetic matrix multiplication of $A$ and $B$
running in $O(n^{2.5})$ time. However, in case
$\lambda(A,\ell,row)=\omega(\sqrt n)$ and
$\lambda(B,k,col))=O(\sqrt n)$ or {\em vice versa} only
Theorem \ref{theo: mmstclus} yields an $O(n^{2.5})$-time algorithm.

\begin{figure}[t]
 \begin{center}
   \includegraphics[width=0.9\textwidth]{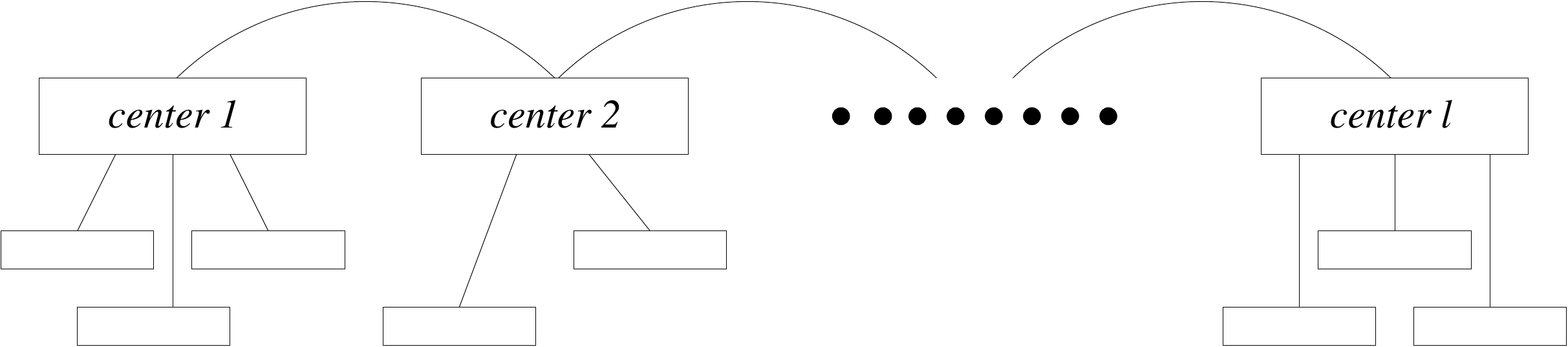}
\end{center}
\caption{The spanning tree $T_A$ of the rows of $A$.}
\label{fig: sp}
\end{figure}}


\section{Potential extensions}

The rows or columns in the input 0-1 matrices can be very long.  Also,
a large number of clusters might be needed in order to obtain
a low upper bound on their radius.  Among other things, for
these reasons, we have picked Gonzalez's classical algorithm for the
$k$-center clustering problem \cite{Gon85} as a basic tool in our
deterministic
approach to the arithmetic matrix product of two 0-1 matrices with
clustered rows or columns.  The running time of his algorithm is
linear not only in the number of input points but also in their
dimension, and in the parameter $k.$ Importantly, it is very simple,
 deterministic,
 provides a solution within $2$ of the optimum,
 and can be applied in Hamming spaces.  For instance, 
there exist  faster (in terms of $n$ and $k$)
$2$-approximation algorithms for
 $k$-center clustering  with hidden exponential dependence on the
 dimension in their running time, see \cite{FG88,HM06}.
 Among several newer works on speeding approximation algorithms
 for $k$-center clustering (e.g., \cite{EHS20,FJLNP25,JKS24,KLP26})
 only \cite{KLP26}
 explicitly includes
 Hamming spaces. The randomized approximation method
 for the $k$-center clustering problem in Hamming spaces from \cite{KLP26}
 is summarized in Fact \ref{fact: 2plus}. 
 It provides approximation guarantees arbitrarily close to $2,$
 is substantially faster than Gonzalez's algorithm when
 the dimension and $k$ are superlogarithmic,
 and is relatively simple. For these reasons, our improved
 randomized algorithms for the approximate and exact 0-1 matrix
 multiplication are based on this method.

One could easily generalize our main results by replacing
approximation algorithms for $k$-center clustering
with approximation algorithms for the more
 general problem of $k$-center clustering with outliers \cite{Char01}.
 In the latter problem, a given number $z$ of input points could be
 discarded as outliers when trying to minimize the maximum cluster
 radius. Unfortunately, the algorithms for this more general
 problem tend to be more complicated and the focus seems to
 be on the approximation ratio achievable in polynomial time
 (e.g., $3$ in \cite{Char01} and $2$ in \cite{lot})
 and not on the time complexity.

There are many other variants of clustering than $k$-center clustering, and
plenty of methods have been developed for them in the literature.
In fact, in the design
of efficient algorithms for the exact arithmetic matrix
product of 0-1 matrices with clustered rows or columns,
using the $k$-median clustering
could seem more natural. The objective in the latter
problem is to minimize the sum of distances
between the input points and their nearest centers.
Unfortunately, no simple deterministic $O(1)$-approximation algorithms
for the latter problem that are  efficient in case the dimension
and $k$ parameters are large seem be to available \cite{char99,chen}.

Our approximate and exact algorithms for the matrix product of 0-1
matrices as well as the preprocessing of the matrices can be
categorized as supervised since they assume that the user
has some knowledge on the input matrices and can choose
reasonable values of the parameters $\ell$ and $k$ guaranteeing
relatively low overall time complexity. Otherwise, one can try
the $\ell$-center and $k$-center clustering subroutines
for a number of combinations of different values of $\ell$
and $k$ in order to pick the combination yielding the lowest
upper bound on the overall time complexity of the algorithm or preprocessing.
The situation is somewhat easier to handle when the input matrices are square,
as shown in the following example.

\medskip 
\noindent
\emph{Example.}
  Let $A$ and $B$ be two 0-1 $n\times n$ matrices.
  The upper bound on the time needed to compute
  the arithmetic product of $A$
  and $B$ in Theorem \ref{theo: mmstclus} and
  Corollary \ref{cor: query}
  simplifies to $O(n^2\min\{\ell+\lambda(A,\ell,row),
  k+\lambda(B,k,col\})).$ Suppose
  that we would like to compute the product in time not exceeding
  roughly $ct(n)$, where $t$ is some nondecreasing positive-integer
  function satisfying $t(n)\ge n^2$
  and $c$ is a small positive constant. This implies that
  both $\ell$ and $k$ should not exceed $n^2/t(n)$ too much.
  On the other hand, both $\lambda(A,\ell,row)$ and $\lambda(B,k,col)$
  are nonincreasing with respect to $\ell$ or $k,$ respectively.
  Hence, we should pick $\ell $ and $k$ as large as possible, e.g.,
  $\lceil t(n)/n^2\rceil .$
  Note that we can compute $2$-approximations of
  $\lambda(A, \lceil t(n)/n^2\rceil, row)$
  and $\lambda(B,\lceil t(n)/n^2\rceil,col)$ by running
  Gonzalez's  algorithm in $O(t(n))$ time.
  If neither $\lambda(A, \lceil t(n)/n^2\rceil, row)$
  nor $\lambda(B,\lceil t(n)/n^2\rceil,col)$ is at most $dt(n)/n^2,$
  where $d$ is some low constant not exceeding $c$, 
  then we need to increase, e.g., double, the threshold function $t(n)$
  and repeat the procedure.


\subsection*{Acknowledgments}
The authors are grateful to the anonymous reviewers for their valuable
comments and suggestions.
Jesper Jansson was partially supported by Japan Society for the Promotion of Science - KAKENHI grant 24K22294.

\bibliographystyle{abbrv}
\bibliography{Bibl_clustering_matr_mult5}

@article{A03,
  author = "Dimitris Achlioptas",
  title = "Database-friendly random projections: {Johnson}-{Lindenstrauss} with
          binary coins",
  journal = "Journal of Computer and System Sciences",
  year = "2003",
  volume = "66",
  number = "4",
  pages = "671--687",
}

@inproceedings{ADV25,
  author  = "J. Alman and R. Duan and V. {Vassilevska Williams} and Y. Xu and
             Z. Xu and R. Zhou",
  title   = "More Asymmetry Yields Faster Matrix Multiplication",
  booktitle = "Proceedings of the Annual ACM-SIAM Symposium on Discrete
               Algorithms (SODA~2025)",
  publisher = "ACM-SIAM",
  address = "Philadelphia, PA",
  pages   = "2005--2039",
  year    = "2025"
}

@inproceedings{AMBFGR_25,
  author  = "J.N.F. Alves and S. Moustafa and S. Benkner and A.P. Francisco and
             W.N. Gansterer and L.M.S. Russo",
  title   = "Accelerating Graph Neural Networks Using a Novel
             Computation-Friendly Matrix Compression Format",
  booktitle = "Proceedings of the 2025 IEEE International Parallel and
               Distributed Processing Symposium (IPDPS~2025)",
  publisher = "IEEE",
  address = "Piscataway, NJ",  
  pages   = "1091--1103",
  year    = "2025"
}

@misc{Emi24,
  author  = "E. Anand and J. van den Brand and R. McCarthy",
  title   = "The Structural Complexity of Matrix-Vector Multiplication",
  howpublished = "arXiv:2502.21240",
  year    = "2025"
}

@inproceedings{AC11,
  author  = "A.~N.~Arslan and A.~Chidri",
  title   = "A Clustering-Based Matrix Multiplication Algorithm",
  booktitle = "Proceedings of the 2011 International Conference on Scientific
              Computing (CSC 2011)",
  publisher = "CSREA Press",
  address = "U.S.A.",
  pages   = "303--307",
  year    = "2011"
}

@inproceedings{AC10,
  author  = "A.~N.~Arslan and A.~Chidri",
  title   = "An efficient multiplication algorithm for
thin matrices and for matrices with similar rows and columns",
  booktitle = "Proceedings of the 2010 International Conference on Scientific
              Computing (CSC 2010)",
  publisher = "CSREA Press",
  address = "U.S.A.",
  pages   = "147--152",
  year    = "2010"
}

@inproceedings{BL,
  author  = "A. Bj{\"o}rklund  and A. Lingas",
  title   = "Fast {Boolean} Matrix Multiplication for Highly Clustered Data",
  booktitle = "Proceedings of the Algorithms and Data Structures Symposium
  (WADS 2001)",
  publisher = "Springer",
  address = "Berlin Heidelberg",
  pages   = "258--263",
  year    = "2001"
}

@inproceedings{Char99,
  author  = "M. Charikar and S. Guha and E. Tardos and D.B. Shmoys",
  title   = "A constant-factor approximation algorithm for the
k-median problem",
  booktitle = "Proceedings of the 31st Annual ACM Symposium on Theory 
of Computing (STOC 1999)",
  publisher = "ACM",
  address = "New York, NY",
  pages   = "1--10",
  year    = "1999"
}

@inproceedings{Char01,
  author  = "M. Charikar and S. Khuller and D.M. Mount and G. Narasimhan",
  title   = "Algorithms for facility location problems with outliers",
  booktitle = "Proceedings of the 12th Annual ACM-SIAM
Symposium on Discrete Algorithms (SODA 2001)",
  publisher = "ACM-SIAM",
  address = "Philadelphia, PA",
  pages   = "642--651",
  year    = "2001"
}

@article{chen,
  author  = "K.~Chen",
  title   = "On Coresets for $k$-Median and $k$-Means Clustering in Metric and {Euclidean} Spaces and Their Applications",
  journal = "SIAM Journal on Computing",
  year    = "2009",
  volume  = "39",
  number  = "3",
  pages   = "923--947"
}

@article{CL99,
  author  = "E.~Cohen and D.~D.~Lewis",
  title   = "Approximating
matrix multiplication for pattern recognition tasks",
  journal  = "Journal of Algorithms",
  year    = "1999",
  volume  = "30 (2)",
  pages   = "211--252"
}

@article{DKM06,
  author  = "P. Drineas and R. Kannan and M.W. Mahoney",
  title   = " Fast {M}onte {C}arlo algorithms for matrices {I}: 
{A}pproximating matrix multiplication",
  journal  = "SIAM Journal on Computing",
  year    = "2006",
  volume  = "36 (1)",
  pages   = "132--157"
}

@article{EHS20,
  author  = "D. Eppstein and S. Har-Peled and A. Sidiropoulos",
  title   = "Approximate greedy
clustering and distance selection for graph metrics",
  journal  = "J. Comput. Geom.",
  year    = "2020",
  volume  = "11 (1)",
  pages   = "629--652"
}

@inproceedings{FG88,
  author  = "T. Feder and D. Greene",
  title   = "Optimal algorithms for approximate clustering",
  booktitle = "Proceedings of ACM Symposium on
Theory of Computing (STOC 1988)",
  publisher = "ACM",
  address = "Chicago, IL",
  pages   = "434--444",
  year    = "1988"
}

@inproceedings{FJLNP25,
  author  = "A. Filtser and S.H.C. Jiang and Y. Li
and A.M. Naredla and I. Psarros and Q. Yang and Q. Zhang",
  title   = "Faster Approximation Algorithms 
for $k$-Center via Data Reduction",
  booktitle = "Proceedings of the 42nd International Conference on Machine Learning",
  series = "PMLR",
  volume = "267",
  pages   = "17189-17202",
  year    = "2025"
}

@article{FKL15,
  author  = "P. Floderus and M. Kowaluk and A. Lingas and E.{-}M. Lundell",
  title   = "Detecting and Counting Small Pattern Graphs",
  journal = "{SIAM} Journal on Discrete Mathematics",
  year    = "2015",
  volume  = "29 (3)",
  pages   = "1322--1339"
}

@article{FJLL18,
  author  = "P. Floderus and J. Jansson and C. Levcopoulos and A. Lingas
and D. Sledneu",
  title   = "{3D} Rectangulations and Geometric Matrix Multiplication",
  journal = "Algorithmica",
  year    = "2018",
  volume  = "80 (1)",
  pages   = "136--154"
}

@article{FL97,
  author  = "M. Frances and A. Litman",
  title   = "On covering problems of codes",
  journal = "Theory of Computing Systems",
  year    = "1997",
  volume  = "30",
  number  = "2",
  pages   = "113--119"
}

@article{GJL04,
author  = "L. {G\k{a}sieniec} and J. Jansson and A. Lingas",
  title   = "Approximation algorithms for {Hamming} clustering problems",
  journal = "{J}ournal of {D}iscrete {A}lgorithms",
  year    = "2004",
  volume  = "2 (2)",
  pages   = "289--301"
}

@article{GKL09,
  author  = "L. {G\k{a}sieniec} and M. Kowaluk and A. Lingas",
  title   = "Faster multi-witnesses for {Boolean} matrix multiplication",
  journal = "{I}nformation {P}rocessing {L}etters",
  year    = "2009",
  volume  = "109 (4)",
  pages   = "242--247"
}

@inproceedings{GL03,
  author  = "L. {G\k{a}sieniec} and A. Lingas",
  title   = "An Improved Bound on {Boolean}
 Matrix Multiplication for Highly Clustered Data",
  booktitle = "Proceedings of the Algorithms and Data Structures Symposium
  (WADS 2003)",
  publisher = "Springer",
  address = "Berlin Heidelberg",
  pages   = "329--339",
  year    = "2003"
}

@article{Gon85,
  author  = "T. Gonzalez",
  title   = "Clustering to minimize the maximum intercluster distance",
  journal = "Theoretical Computer Scirnce",
  year    = "1985",
  volume  = "38",
  pages   = "293--306"
}

@article{HM06,
  author  = "S. Har-Peled and M. Mendel",
  title   = " Fast construction of nets in low-dimensional 
metrics and their applications",
  journal = "SIAM Journal on Computing",
  year    = "2006",
  volume  = "35",
  number  = "5",
  pages   = "1148--1184"
}

@article{lot,
  author  = "D.G. Harris and T. Pensyl and A. Srinivasan and K. Trinh",
  title   = "A Lottery Model for Center-Type Problems With Outliers",
  journal = "ACM Transactions on Algorithms",
  year    = "2019",
  volume  = "15",
  number  = "3",
  pages   = "1--25",
  note    = "Article No.~36."
}

@inproceedings{JKS24,
  author = "S.H.C. Jiang and R. Krauthgamer and S. Sapir",
  title = "Moderate Dimension Reduction for $k$-Center Clustering",
  booktitle = "Proceedings of the International Symposium on
Computational Geometry (SoCG 2024)",
  publisher = "Schloss Dagstuhl -- Leibniz-Zentrum {f\"{u}r} Informatik GmbH,
Dagstuhl Publishing",
  address = "{Saarbr\"{u}cken/Wadern}, Germany",
  series = "LIPIcs",
  volume = "293",
  pages   = "64:1--64:16",
  year    = "2024"
}

@inproceedings{IM,
  author  = "P. Indyk and R. Motwani",
  title   = "Approximate Nearest Neighbors: Towards Removing the Curse of
Dimensionality",
  booktitle = "Proceedings of ACM Symposium on
Theory of Computing (STOC 1998)",
  publisher = "ACM",
  pages   = "604--613",
  year    = "1998"
}

@article{IS09,
  author  = "M.A. Iven and C.V. Spencer",
  title   = "A note on compressed sensing and the complexity of matrix
multiplication",
  journal = "Information Processing Letters",
  year    = "2009",
  volume  = "109",
  number  = "10",
  pages   = "468--471"
}

@inproceedings{JL84,
  author  = "W.B. Johnson and J. Lindenstrauss",
  title   = "Extensions of {Lipschitz} mappings into a {Hilbert} space",
  booktitle = "Conference in Modern Analysis and Probability",
  publisher = "American Mathematical Society",
  address = "Providence, RI",
  series = "Contemporary Mathematics",
  volume = "26",
  pages   = "189--206",
  year    = "1984"
}

@article{KLP26,
  author  = "M. Kowaluk and A. Lingas and M. Persson",
  title   = "Fast approximate $\ell$-center clustering in high dimensional space",
  journal = "Algorithms MDPI",
  year    = "2026",
  volume  = "19",
  number  = "3",
  pages   = "243"
}

@inproceedings{K13,
  author  = "K. Kurzkov",
  title   = "Deterministic algorithms for skewed matrix products",
  booktitle = "Proceedings of the International Symposium on Theoretical Aspects of Computer Science (STACS 2013)",
  publisher = "Schloss Dagstuhl -- Leibniz-Zentrum {f\"{u}r} Informatik GmbH,
Dagstuhl Publishing",
  address = "{Saarbr\"{u}cken/Wadern}, Germany",
  series = "LIPIcs",
  volume = "20",
  pages   = "466--477",
  year    = "2013"
}

@article{P13,
  author  = "R.~Pagh",
  title   = "Compressed matrix multiplication",
  journal = "ACM Transactions on Computation Theory (TOCT)",
  year    = "2013",
  volume  = "5",
  number  = "3",
  pages   = "1--17"
}

@inproceedings{R,
  author  = "P.~Robinson",
  title   = "Brief Announcement: What Can We Compute in a Single Round of the
             Congested Clique?",
  booktitle = "Proceedings of the 2023 ACM Symposium on Principles of
               Distributed Computing (PODC~2023)",
  publisher = "ACM",
  pages   = "168--171",
  year    = "2023"
}

@inproceedings{S06,
  author  = "T. Sarl\'os",
  title   = "Improved approximation algorithms for large matrices via random projections",
  booktitle = "Proceedings of the IEEE Symposium on Foundations of Computer Science (FOCS 2006)",
  publisher = "IEEE Computer Society",
  address = "Piscataway, NJ",  
  pages   = "143--152",
  year    = "2006"
}

@inproceedings{VXZ24,
  author  = "V. {Vassilevska Williams} and Y. Xu and Z. Xu and R. Zhou",
  title   = "New Bounds for Matrix Multiplication: from Alpha to Omega",
  booktitle = "Proceedings of the Annual 
ACM-SIAM Symposium on Discrete
               Algorithms (SODA~2024)",
  publisher = "ACM-SIAM",
  pages   = "3792--3835",
  address = "Philadelphia, PA",
  year    = "2024"
}

@article{YZ,
  author  = "R. Yuster and U. Zwick",
  title   = "Fast sparse matrix multiplication",
  journal = "ACM Transactions on Algorithms",
  year    = "2005",
  volume  = "1",
  pages   = "2--13"
}


\end{document}